\documentclass[twocolumn, aps, pra, showpacs, superscriptaddress, floatfix, nofootinbib, 10pt]{revtex4-1}
\usepackage{amsmath,bbm}
\usepackage{amssymb}
\usepackage{graphicx}
\usepackage{color}
\usepackage{amsthm}
\usepackage{ dsfont }
\usepackage{float}
\usepackage{cancel}
\usepackage[autostyle]{csquotes}

\usepackage[colorinlistoftodos]{todonotes}
\usepackage[colorlinks=true, allcolors=blue]{hyperref}
\usepackage{amsthm}
\usepackage{amsmath}
\usepackage{amssymb}
\usepackage{graphicx}
\usepackage{color}
\usepackage{enumitem}   

\newtheorem{theorem}{Theorem}
\newtheorem{definition}{Definition}
\newtheorem{lemma}{Lemma}

\newtheorem{proposition}{Proposition}

\newtheorem*{theorem*}{Theorem}

\newcommand{\SSS}{\mathcal{S}}

\newcommand{\NL}{\mathcal{N}_\mathcal{L}}
\newcommand{\NQ}{\mathcal{N}_\mathcal{Q}}

\newcommand{\NB}{\mathcal{N}_\mathcal{B}}
\newcommand{\fR}{\mathfrak{R}}

\newcommand{\beq}{\begin{equation}}
\newcommand{\eeq}{\end{equation}}
\newcommand{\beqa}{\begin{eqnarray}}
\newcommand{\eeqa}{\end{eqnarray}}
\newcommand{\bra}[1]{\ensuremath{\left\langle#1\right|}}
\newcommand{\ket}[1]{\ensuremath{\left|#1\right\rangle}}

\newcommand{\mean}[1]{\mathds{E}\left[#1\right]}
\newcommand{\dd}{\mathrm{d}}

\newcommand{\norm}[1]{\left\lVert#1\right\rVert}

\renewcommand{\today}{\number\day\space\ifcase\month\or
   January\or February\or March\or April\or May\or June\or
   July\or August\or September\or October\or November\or December\fi
   \space\number\year}

\begin{document}

\title{Limits on correlations in networks for quantum and no-signaling resources}

\date{\today}

\author{Marc-Olivier Renou}
\affiliation{D\'epartement de Physique Appliqu\'ee, Universit\'e de Gen\`eve, CH-1211 Gen\`eve, Switzerland}

\author{Yuyi Wang}
\affiliation{Distributed Computing Group, ETH Zurich, Switzerland} 

\author{Sadra Boreiri}
\affiliation{D\'epartement de Physique Appliqu\'ee, Universit\'e de Gen\`eve, CH-1211 Gen\`eve, Switzerland}

\author{Salman Beigi}
\affiliation{School of Mathematics, Institute for Research in Fundamental Sciences (IPM), Tehran,
Iran}

\author{Nicolas Gisin}
\affiliation{D\'epartement de Physique Appliqu\'ee, Universit\'e de Gen\`eve, CH-1211 Gen\`eve, Switzerland}

\author{Nicolas Brunner}
\affiliation{D\'epartement de Physique Appliqu\'ee, Universit\'e de Gen\`eve, CH-1211 Gen\`eve, Switzerland}

\begin{abstract}
A quantum network consists of independent sources distributing entangled states to distant nodes which can then perform entangled measurements, thus establishing correlations across the entire network. But how strong can these correlations be? Here we address this question, by deriving bounds on possible quantum correlations in a given network. These bounds are nonlinear inequalities that depend only on the topology of the network. We discuss in detail the notably challenging case of the triangle network. Moreover, we conjecture that our bounds hold in general no-signaling theories. In particular, we prove that our inequalities for the triangle network hold when the sources are arbitrary no-signaling boxes which can be wired together. Finally, we discuss an application of our results for the device-independent characterization of the topology of a quantum network.
\end{abstract}
\maketitle

\section{Introduction}Quantum nonlocality, i.e., the fact that distant observers performing local measurements on a shared entangled quantum state can violate a Bell inequality, is a key feature of quantum theory \cite{bell}. In recent years, considerable efforts have been devoted, both theoretically and experimentally, to deepen our understanding of this phenomenon \cite{review}. Of particular interest is the investigation of quantum nonlocality in the context of general networks \cite{branciard,branciard2,fritz}. Here, a set of distant observers share entanglement distributed by several sources which are assumed to be independent from each other. As each source distributes entanglement to only certain subsets of observers, new limits on possible correlations arise. Moreover, observers can correlate particles coming from different independent sources (e.g. via entangled quantum measurements, as in quantum teleportation \cite{bennett}), and thus generate strong correlations across the entire network. Notably, this leads to astonishing new effects, such as the possibility of violating a Bell inequality without the need for inputs \cite{fritz,branciard2}. Beyond the fundamental interest, these ideas are also directly relevant to the development of real-world quantum networks \cite{kimble,simon}.

It is fair to say, however, that our understanding of quantum nonlocality in networks is still very limited \cite{NG2018}. A first challenge is to characterize classical correlations in networks, i.e. when all sources distribute only classical variables. Due to the independence condition of the sources, the set of classical correlations is no longer convex (contrary to the standard Bell scenario, featuring a single common source, see e.g. \cite{review}). Therefore, relevant Bell inequalities must be nonlinear. Examples of such inequalities have been derived (see e.g. \cite{branciard,branciard2,chaves2,armin,rosset,chaves,wolfe,luo}), but the general structure of this problem is still not understood.

Another challenge, which represents the starting point of this work, is to understand the limits of quantum correlations in networks. Specifically, given a certain network, we aim at determining fundamental constraints on achievable correlations when using quantum resources. Hence we consider any possible quantum strategy compatible with the network topology. This involves sources producing arbitrary quantum states (of any Hilbert space dimension), and nodes performing arbitrary joint quantum measurements. Interestingly it turns out that fundamental limits arise here even without involving any inputs, in contrast to the standard Bell scenario. This means that the network topology imposes fundamental limitations on achievable correlations.

We start our investigation with the case of networks featuring three observers, each of them providing an output (but receiving no input). We discuss in detail the notably challenging case of the ``triangle network'' \cite{branciard2,fritz,gisin,wolfe}, for which we identify a nonlinear inequality capturing partly the set of quantum correlations. More generally we derive a family of inequalities satisfied by quantum correlations considering an arbitrary network with bipartite sources. Interestingly these inequalities can be viewed as quantum versions of the Finner inequalities \cite{finner}, introduced in a completely different context, namely graph theory.

An interesting application of our inequalities is that they allow one to test the topology of an unknown quantum network in a device-independent manner. That is, by simply considering the observed correlations, one can tell whether a certain network topology is compatible or not. If the observed data violates one of our inequalities, then the corresponding network topology can be ruled out immediately.

Finally, we go beyond quantum correlations, and consider more general no-signaling resources \cite{PR,barrettPR}.
That is, each source now distributes a no-signaling (NS) box, and each node performs a joint operations on these resources \cite{barrett}. Notably, we show that, for the triangle network, the nonlinear inequality we obtained for quantum correlations also holds in a general no-signaling theory, where each source can produce an arbitrary number of bipartite NS boxes, and each node performs an arbitrary wiring on the received resources. This leads us to the conjecture that Finner inequalities captures in fact the limit of correlations in networks for any possible no-signaling theory. It thus represents a general limit of achievable correlations in networks, independently of the underlying physical model, given the latter does not allow for instantaneous communication (in other words, is compatible with special relativity).

\section{Three-observer networks}
\begin{figure}
\includegraphics[scale=0.2]{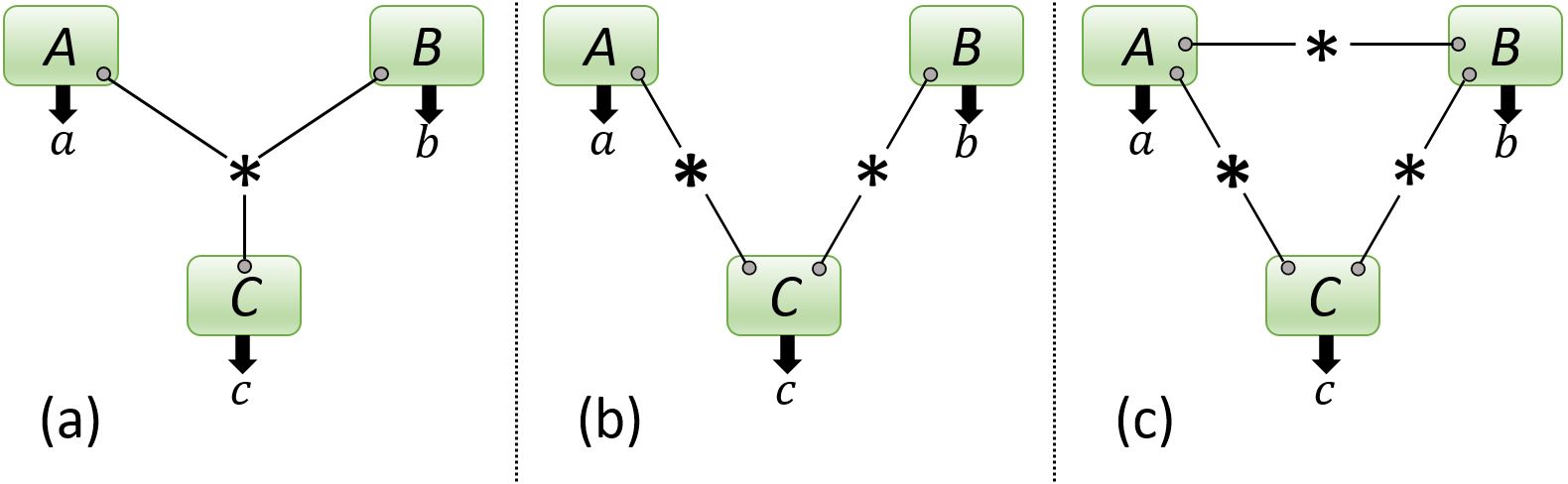}
\caption{All inequivalent three-party networks.}
\label{threeparties}
\end{figure}
Consider a quantum network with three observers $A$, $B$ and $C$, featuring one (or more) sources, distributing quantum states to subset of the parties. Each party then performs a measurement on the received quantum systems, leading to outputs denoted $a$, $b$ and $c$. 
Here for simplicity we assume that $a, b, c\in \{0,1\}$ are binary, but later consider larger output sets.

One can consider three inequivalent networks here. The first, depicted in Fig.~\ref{threeparties}(a), features a single common source distributing a quantum state to all three observers. This corresponds to the situation considered in the standard Bell scenario (see e.g. \cite{review}), except that parties receive no inputs in our case. It is straightforward to see that any possible distribution $P(abc)$ can be achieved. In fact, it is enough to restrict to classical sources here. The source samples from the distribution $P(abc)$, and then distributes the obtained outputs to each observer. Geometrically, the set of possible attainable distributions $P(abc)$  is nothing but the whole probability simplex, which is a 7-dimensional simplex in $\mathbb R^8$ due to the normalization constraint $\sum_{a,b,c} P(abc) =1$.

A more interesting scenario is when the network features two independent sources, as in Fig.~\ref{threeparties}(b). The first source distributes a common state to $A$ and $B$, and the second independent source to $B$ and $C$. This scenario is known as bilocality \cite{branciard}, and corresponds to the setup of entanglement swapping. In this case, the parties $A$ and $C$ are initially independent, and can only be correlated via $B$. Hence, if one traces out $B$, the marginal statistics of $A$ and $C$ must factorize. We have the causality condition:
\begin{equation} \label{biloc}
\sum_b P(abc) = P_{AC}(ac) = P_A(a) P_C(c) \,.
\end{equation}
Hence, contrary to the first network discussed above, not all correlations are possible in the bilocality network. It turns out however that the constraint \eqref{biloc} is enough to characterize achievable correlations: any $P(abc)$ satisfying \eqref{biloc} can be achieved. 
It is again enough to consider only classical variables. 
Specifically, let the first (resp. second) source sample from $P_A(a)$ (resp. $P_C(c)$) and distribute the output to $A$ and $B$ (resp. $B$ and $C$) and $B$ use local randomness to sample $P(b|ac)$. 
Geometrically, the set of achievable distributions $P(abc)$ forms a 6-dimensional curved manifold in $\mathbb{R}^8$.

Next we move to the third---and arguably the most interesting and challenging---configuration, i.e. the triangle network (see Fig.~\ref{threeparties}(c)).
Consider the bilocality network again, and add a source connecting $A$ and $C$. Due to this additional source, the independence condition \eqref{biloc} does no longer hold.
In fact, one can show that, besides the normalization constraint, there is no other equality constraint for this network (which would reduce the dimension of the set).
This follows from the fact that the maximally mixed (uniform) distribution $P_u(abc) = 1/8$ $\forall a,b,c$ is surrounded by a ball of achievable distributions; see Appendix~\ref{cube_representation}.

It turns out, however, that not all distributions $P(abc)$ are achievable in the triangle scenario, as shown in Appendix~\ref{cube_representation} and Ref.~\cite{wolfe} via specific examples. Thus, the set of possible distributions forms a strict subset of the probability simplex, yet its characterization is a challenging problem. Here we derive a relevant nonlinear inequality that necessarily holds in quantum theory. 

\begin{theorem}\label{Finner_triangle-1}
In the triangle network (Fig.~\ref{threeparties}(c)), quantum correlations necessarily satisfy
\begin{equation}\label{Finner_triangle}
P(abc) \leq \sqrt{P_A(a)P_B(b)P_C(c)} \,.
\end{equation}
\end{theorem}

As quantum correlations are stronger than classical ones, inequality \eqref{Finner_triangle} also holds for the case where the sources emit classical variables: in this case Theorem~\ref{Finner_triangle-1} can be derived by applying two Cauchy-Schwarz inequalities on $\mathbb E[f_A g_B h_C]$ where $f_A, g_B$ and $h_C$ are the characteristic functions of the sets $\{a\}, \{b\}$ and $\{c\}$ respectively. Back to quantum sources, Theorem~\ref{Finner_triangle-1} can be proven by essentially the same ideas, but as we will later prove a generalization of this theorem, we skip the proof here.

In the classical case, the set of all possible strategies can be understood intuitively in geometrical terms, as a 3-dimensional cube. In this case, the inequality~\eqref{Finner_triangle} follows from the Loomis-Whitney inequality, capturing the fact that the volume of a 3-dimensional object is upper bounded by the product of the areas of the object's projections in three orthogonal directions (see Appendix~\ref{cube_representation}). 

The inequality \eqref{Finner_triangle} allows us to prove that a large range of distributions cannot be achieved in the quantum triangle network. Consider for instance the family of distributions

\begin{equation}\label{pq}
P_{p,q} = p \delta_{000} + q \delta_{111} + (1-p-q) P_{\text{diff}}
\end{equation}
where $\delta_{abc}$ represents the distribution that always outputs $a$, $b$ and $c$ deterministically, and $P_{\text{diff}}$ is the uniform distribution over $\{0,1\}^3\setminus\{000, 111\}$, i.e., $P_{\text{diff}}= (\delta_{001} + \delta_{010}+ \delta_{100}+ \delta_{011}+ \delta_{101}+ \delta_{110})/6$. From inequality \eqref{Finner_triangle} it follows that $P_{p,q} $ is not realizable in the triangle network when $q> 1+p -2p^{2/3}$; see Fig.~\ref{pq_region}. This also shows that the set of quantum distributions achievable in the triangle network is not star convex \footnote{Note that here we show that the set is not star convex with respect to the identity, which implies that the set is not star convex in general (invoking symmetry arguments and the fact that if a set is star convex with respect to two points, then it is star convex with respect to any points between those two).} since distributions of the form $r \delta_{111}+ (1-r) P_u$ violate inequality~\eqref{Finner_triangle} when $7/8 <r<1$; $P_u$ denotes the uniform distribution over $\{0,1\}^3$.

The above example also illustrates how our results can be used to test the topology of an initially unknown network. Suppose Alice, Bob and Charlie observe a distribution $P(abc)$ that violates inequality \eqref{Finner_triangle}. Then, they can certify that the underlying network is not of the triangle type (neither bilocal indeed), but must feature a common source distributing information to all three parties. Note that this test is device-independent, as it is based only on the observed data $P(abc)$. 

\begin{figure}
\includegraphics[scale=0.23]{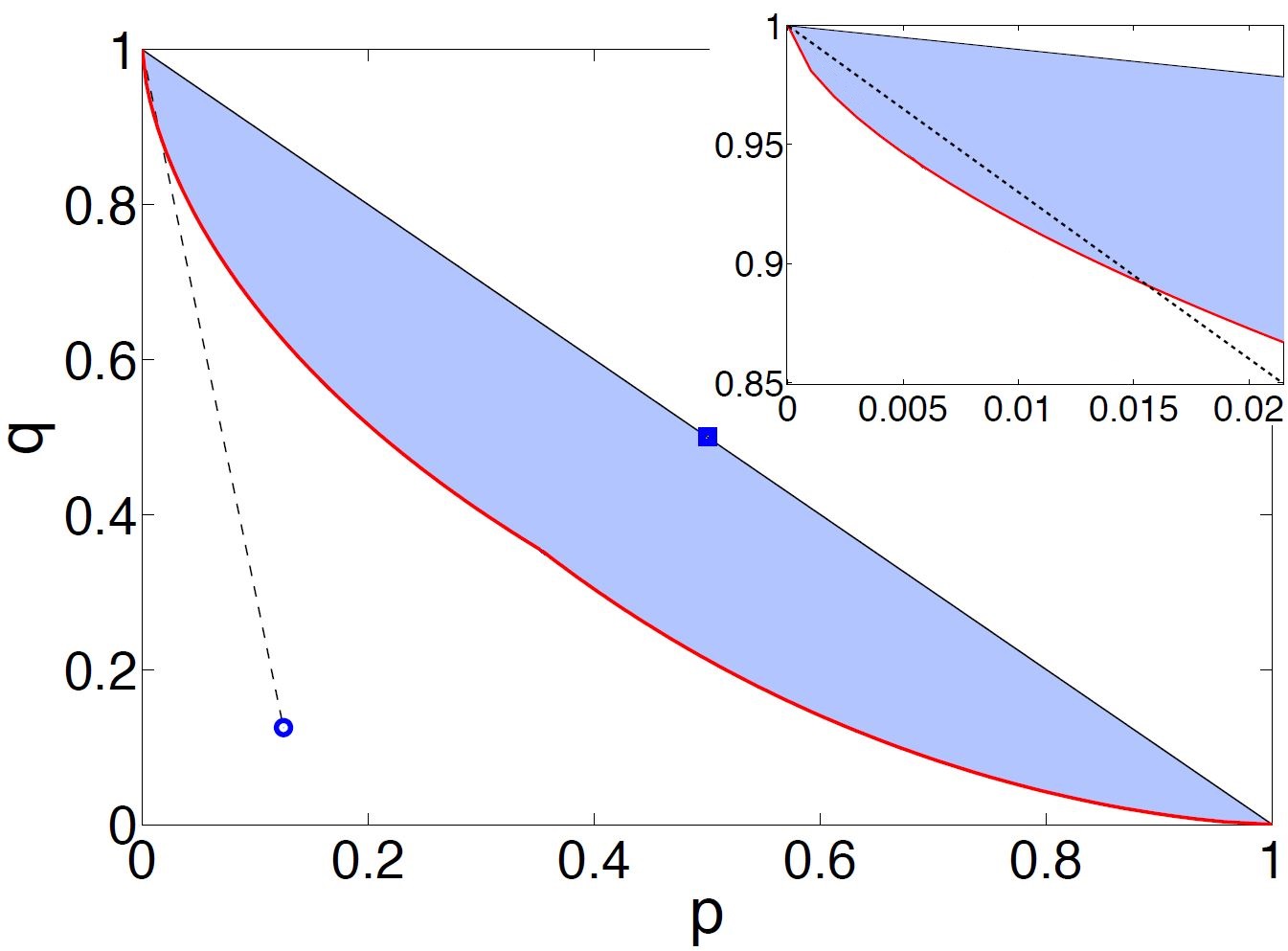}
\caption{Geometrical representation of the set of distributions $P_{p,q}$ of Eq. \eqref{pq}. Distributions in the shaded area are not achievable in the triangle network in quantum theory, as they do not satisfy inequality \eqref{Finner_triangle}. Distributions below the red curve satisfy the inequality, and are thus potentially achievable in quantum mechanics. The blue circle represents the maximally mixed distribution $P_u$, while the blue square is the so-called GHZ distribution $p=q=1/2$. Distributions of the form $r \delta_{111}+ (1-r) P_u$ (dashed line) violate inequality \eqref{Finner_triangle} for $7/8<r<1$ (see inset), showing that the quantum set is not star convex.} 
\label{pq_region}
\end{figure}

In the remainder of the paper, we will generalize Theorem~\ref{Finner_triangle-1} in two different directions. First, we will show how to derive similar nonlinear inequalities for a larger class of networks. Second, we will prove that inequality \eqref{Finner_triangle} holds also in the triangle network for a generalized probabilistic theory.

\section{General networks}
We now consider networks with an arbitrary number of parties (outputs) yet we mostly restrict to bipartite sources. That is, we assume that our networks $\mathcal N$ consist of $n$ parties $A_1, \dots, A_n$, and an arbitrary number of sources each of which is connected to a pair of parties. Thus a network $\mathcal N$ can be thought of as a graph over $n$ vertices whose edges represent sources (e.g., the triangle network is represented by the triangle graph).   
The following theorem presents a generalization of Theorem~\ref{Finner_triangle-1} for arbitrary graphs whose proof is given in Appendix~\ref{finner_quantum}.

\begin{theorem}\label{thm:quantum-finner} (Quantum Finner inequality) Consider a network $\mathcal{N}$ with $n$ observers $A_1, \dots, A_n$ and some bipartite sources. Let $\eta=\{\eta_1, ..., \eta_n\}$ be a fractional independent set of $\mathcal{N}$, i.e. weights $\eta_j$ attributed to $A_j$'s are such that, for each source the sum of weights of parties connected to it is smaller than or equal to 1.  \footnote{Note that here we do not use the classical definition of (fractional) independent sets in hypergraph theory.}
Let $f_{j}$ be any real positive local post-processing (function) of the classical output of party $A_{j}$. 
Then, any distribution $P$ achievable in quantum network $\mathcal N$ satisfies
\begin{equation}\label{NQFinner}
\mean{\prod_j f_j} \leq \prod_j \norm{f_j}_{1/\eta_j},
\end{equation}
where $\norm{f}_{1/\eta}=(\mean{ f^{1/\eta}})^{\eta}$ and the expectations are with respect to $P$.
In particular, letting $f_{j}$ being the indicator function of the output of $A_j$ being $a_j$, we have
\begin{equation}\label{NQFinnerProba}
P(a_1...a_n)\leq \prod_{j=1}^n \left(P_{A_{j}}(a_j)\right)^{\eta_j}.
\end{equation}
\end{theorem}
Inequality \eqref{NQFinner} has been derived by Finner~\cite{finner} in the context of graph theory, as a generalization of H\"older's inequality. In our setting, the proof of Finner directly applies to arbitrary networks with multipartite classical sources. Our theorem here generalizes Finner's inequality for arbitrary networks with bipartite sources that are quantum. 

Note that although Finner's inequality is presented as a continuous family of inequalities depending on the choice of weights $(\eta_1, \dots, \eta_n)$, it can be reduced to a finite set of inequalities for a given network (see Appendix~\ref{app:HR}). In particular, for the triangle network, the only nontrivial fractional independent set corresponds to $ \boldsymbol\eta =(1/2, 1/2, 1/2)$, in which case \eqref{NQFinnerProba} reduces to~\eqref{Finner_triangle}.

\section{Triangle network with no-signaling boxes}
We now consider the correlations achievable in the triangle network in a generalized no-signaling theory~\cite{barrett}. In this model, the resources are not quantum states, but general NS boxes. In the triangle network, each source can thus distribute some NS boxes to the two parties connected to it. We emphasis that
in general, each source can distribute several NS boxes that are not necessarily identical and can have arbitrary number of inputs and outputs.
These NS boxes thus serve as a resource for the parties to generate correlated outputs; each party having an arbitrary number of NS boxes shared with others, can locally ``wire'' these boxes in the most general way to determine an output\footnote{Note that this model corresponds essentially to the generalized probabilistic theory of ``boxworld'' \cite{barrett}, except for the fact that boxworld allows for certain multipartite effects (i.e., measurements) that are not wirings \cite{short_barrett}.}. That is, the inputs of certain boxes can be chosen by the party, while others can be determined by wirings, the output of one box being used as in the input for another one. We can prove the following result. 

\begin{theorem}\label{Finner_boxworld}
Suppose that a distribution $P$ is achievable in the triangle network when the sources distribute arbitrary NS boxes and the parties perform arbitrary local wirings. Then $P$ satisfies~\eqref{Finner_triangle}.
\end{theorem}

Clearly, this result does not follow from Theorem~\ref{Finner_triangle-1}, as here the sources can distribute stronger nonlocal resources than what is possible in quantum theory~\cite{PR}. However, we also point out that Theorem~\ref{Finner_boxworld} does not imply Theorem~\ref{Finner_triangle-1}, as squantum theory allows for joint entangled measurements which cannot be described as wirings and admit no equivalent in general no-signaling theories~\cite{barrett,short1}.

Here we give the proof ingredients while all details are left for Appendix~\ref{finner_boxworld_proof}.
A key point in the proof is the notion of the Hypercontractivity Ribbon (HR) studied in~\cite{beigi} as a monotone measure of non-local correlations. 

\begin{definition}[Hypercontractivity Ribbon]\label{hypercontractivity_ribbon}
The HR $\fR(P)$ of a tripartite distribution $P(abc)$ is the set of non-negative triplets $(\alpha, \beta, \gamma)$ such that for any real functions $f(a)$, $g(b)$ and $h(c)$ of the outputs $a$, $b$, $c$, we have 
\begin{equation}\label{hypercontractivity_ribbon_2}
\mean{fgh}\leq  \mean{|f|^{1/\alpha}}^\alpha \mean{|g|^{1/\beta}}^\beta \mean{|h|^{1/\gamma}}^\gamma \,.
\end{equation}
\end{definition}

By H\"older's inequality, the HR of the maximally mixed distribution $P_u$ (the weakest resource for establishing correlations) is the entire unit cube. Also it is not hard to verify that the HR of the GHZ distribution $(\delta_{000}+\delta_{111})/2$ which is the best possible resource, is the half cube given by the vertices $(0,0,0), (1,0,0), (0,1,0), (0,0,1)$.

The main feature of HR is its monotonicity under local operations. That is, if a tripartite distribution $Q(a'b'c')$ can be obtain by local post-processing of outcomes of another distribution $P(abc)$, then $\fR(P)\subseteq \fR(Q)$. In particular, with the GHZ distribution we can simulate $P_u$, but no the other way around.

Now with the definition of HR in hand, equation~\eqref{Finner_triangle} essentially says that $(1/2,1/2,1/2)\in\fR(P)$.
In Appendix~\ref{app:HR}, we give an alternative characterization of HR in terms of mutual information which allows us to prove Theorem~\ref{Finner_boxworld}.

\section{Discussion}
We have presented fundamental constraints on quantum correlations achievable in networks. The constraints take the form of nonlinear inequalities, that can be viewed as the quantum version of the Finner inequalities. In particular, we have discussed in detail the case of the triangle network, as well as the problem of device-independently testing the topology of an unknown quantum network.

A natural question is indeed whether the quantum Finner inequalities fully capture the set of quantum correlations in networks. This appears not to be the case in general. Indeed, for the triangle network, there exist correlations that are provably not achievable in quantum theory that do not violate our inequality \eqref{Finner_triangle} \footnote{The so-called ``W'' distribution, $P_W =( \delta_{001}+\delta_{010}+\delta_{100})/3$, cannot be done with quantum resources, which can be proven using the inflation technique of~\cite{wolfe}.}. It would be interesting to derive other forms of constraints. A possibility in this direction would be to exploit the ``reverse Holder'' inequality, the quantum version of which can be straightforwardly derived. Whether this new inequality will turn out to be stronger than the ones we presented is not clear.  More generally, one should generalise the quantum Finner inequalities to the case of sources producing multipartite quantum states.

Finally, we also discussed the limits of correlations in networks when considering theories beyond quantum mechanics. In particular, we could show that inequality \eqref{Finner_triangle} also holds for the triangle network with generalized no-signaling resources. More generally, we conjecture that the inequality \eqref{NQFinner} holds for any no-signaling theory, for more general networks. Note that for the triangle network, this does not follow from our results, as each generalized probabilistic theory features its own set of allowed no-signaling correlations and the set of allowed joint measurements; the two sets being dual to each other \cite{barrett,short_barrett}. If our conjecture is correct, this means that the Finner inequalities capture the limits of achievable correlations in a network, that must hold in any no-signaling theory; the Finner inequalities could thus be viewed as a generalisation of the standard no-signaling condition to networks.


\emph{Acknowledgements.---}We thank Denis Rosset, Armin Tavakoli, Amin Gohari and Elie Wolfe for discussions. We acknowledge financial support from the Swiss national science foundation (Starting grant DIAQ, NCCR-QSIT).

\clearpage

\appendix

\section{Alternative description of networks}\label{network_formalism}





A network consists of a pair $(\mathcal A, \mathcal S)$ where $\mathcal A=\{A_1, \dots, A_n\}$ is the set of parties and $\mathcal S=\{S_1, \dots, S_m\}$ is the set of sources each of which is  shared among a specific subset of parties. Each party $A_j$ produces an output after a post-processing her available sources.

Here we do not limit the amount of information provided by the sources, but restrict ourselves to distributions in which the output of each party $A_j$ has a finite alphabet set.
We also restrict ourselves to minimal networks, where no source is connected to a subgroup of parties already connected by another source. 
However, as the amount of information provided by the source is not restricted, we can always (for convenience) add additional sources shared by parties already connected by another source: in Appendix~\ref{cube_representation}, we sometimes assume that the parties have there own source of randomness.

Since here we are mostly interested in bipartite sources we may think of the network as a graph.
We may think of $\mathcal N$ as a graph whose vertices are labeled by $A_j$'s and whose edges are labeled by $S_i$'s. That is, for any $i$ there is an edge $e_i$ of the graph that connects the two parties that share the source $S_i$. For later use, we adopt the notation $S_i\rightarrow A_j$ (or simply $i\rightarrow j$) to represent that $A_j$ receives a share from the source $S_i$, i.e., $A_j$ is connected to the edge $e_i$.


\subsection{Classical variable models}\label{formalism_HVM}

In the classical case we assume that the sources $S_i$ share randomness among the parties, and each party $A_j$  applies a function on the receives randomnesses to determine her output. 
As proven in \cite{rosset} in this case, we can assume that each source $S_i$ only takes a finite number of values.
Alternatively, we can suppose that $S_i$ is associated to a uniform random variable $s_i\in [0,1]$ that is discretized by a single step function over $[0,1]$ by the parties.
Here we adopt the latter notation. 
Then the party $A_j$ applies a function on the sources she receives, i.e., on $\{s_i: i\rightarrow j\}$, and outputs $a_j$. Let us denote by $r^j_a$ the indicator function that $a_j$ equals $a$, i.e., for a given realization of the sources $\{s_i\}$ we define $r^j_a(\{s_i: i\to j\})$ to be equal to $1$ if the output of $A_j$ given inputs $\{s_i: i\to j\}$ equals $a$, and $0$ otherwise. 
Then the joint output distribution can be written as
\begin{equation}
P(a_1, \dots, a_n)=\int \Big( \prod_j r^{j}_{a_j} \big(\{s_i: i\to j\}\big)\Big)\prod_i \dd s_i,
\end{equation} 
and the marginals are given by
\begin{equation}
P_{A_{j}}(a_j)=\int r^{j}_{a_j} \big(\{s_i: i\to j\}\big)\prod_{i: i\rightarrow j} \dd s_i .
\end{equation} 
We denote the set of all distributions achievable in the classical variable model for a given network $\mathcal N$ by $\mathcal N_{\mathcal L}$.

\subsection{Quantum models}
In the quantum model we assume that sources $S_i$ share quantum states, and the parties determine their outputs by applying measurements. 
As we have no dimension restriction, we can assume that the shared states are pure and the measurements are projective. 
We denote by $\rho_i$ the quantum state distributed by source $S_i$ and write $\rho=\bigotimes_i\rho_i$. Moreover, we denote the measurement operators of $A_j$ by $\big\{M^{(j)}_{a_j}\big\}$. Then the resulting output distribution equals
$$P(a_1, \dots, a_n) = \text{Tr}\bigg( \rho \cdot \bigotimes_j M^{(j)}_{a_j}       \bigg).$$
We denote the set of all distributions achievable in the quantum model for a given network $\mathcal N$ by $\mathcal N_{\mathcal Q}$.


\subsection{Boxworld model}

In this model, each source distributes an arbitrary number of NS boxes to the parties to which it connects.
These NS boxes serve as a resource for the parties to generate correlated outputs.
Each party can locally ``wire'' her boxes in hand to determine her output.
That is, each party successively choose a box and its input, and receives an output. 
She is free to choose the order in which she uses her boxes and the input of it with some stochastic post-processing of her transcript at that time, i.e., previous choices of boxes, their inputs and their outputs.
The final output of each party is a stochastic post-processing of her final transcript. We will later formalize this definition in a more precise way. 
We denote the set of all distributions achievable in the boxworld model for a given network $\mathcal N$ by $\mathcal N_{\mathcal B}$.

\subsection{Fractional independent sets}

Later we will use the following definition. 

\begin{definition}
A fractional independent set of a network $\mathcal{N}$
is a vector of non-negative $\boldsymbol\eta=(\eta_1, ..., \eta_n)$ which corresponds a weight each party such that for each source the summation of the weights of parties connected to it at most 1.
Formally, $\eta_j\geq 0$ for all $j$ and
\begin{equation}\label{fractionnal_matchings}
\sum_{j:\, i\rightarrow j} \eta_j \leq 1, \qquad \forall i.
\end{equation}
\end{definition}

A vector $\boldsymbol \eta$ satisfying the above conditions is called a fractional independent set since assuming that $\eta_j$'s are either $0$ or $1$, the subset $\{A_j:\, \eta_j=1\}$ forms an independent set of the associated graph, i.e., a subset of vertices to two of which are adjacent.

\section{Basic properties in the triangle scenario}\label{cube_representation}

In the following, we present some basic results about correlations achievable in the triangle network. 

\subsection{Cube representation of strategies in the triangle scenario}

\begin{figure}
\center
\includegraphics[scale=0.38]{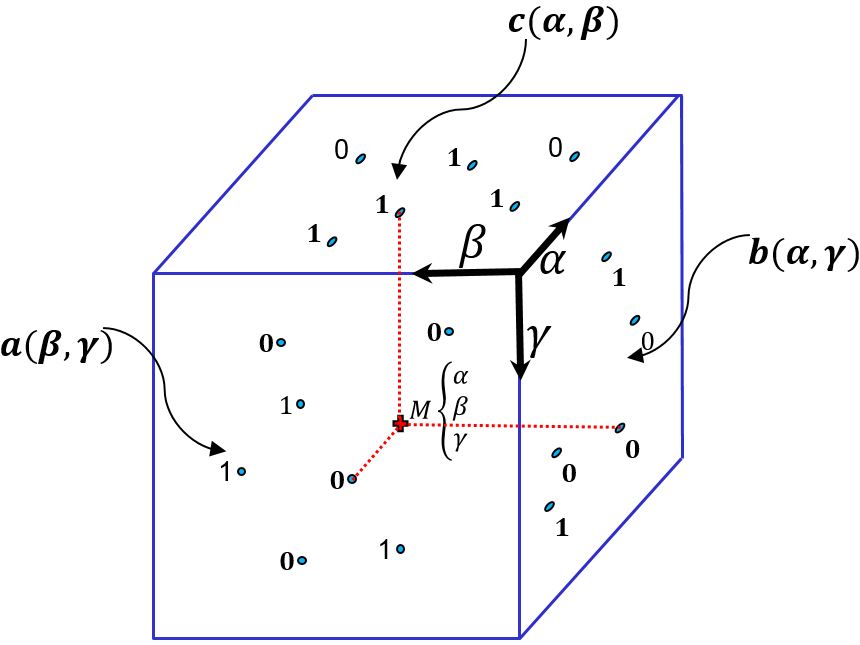}
\caption{Any local strategy in the triangle scenario can be mapped to a unit cube. Three orthogonal axis of the cube are labeled by the hidden variables $\alpha,\beta,\gamma$. Alice's response for a given ($\beta,  \gamma$) is written on the face orthogonal to the $\alpha$ direction (and similarly for Bob and Charlie).
In this representation, $P(000)$ is the volume of points which project to $0$ on all the three faces, and $P_A(0)$ is the area of points one the face orthogonal to the $\alpha$ direction which are labeled $0$.} 
\label{strat_triangle_appendix}
\end{figure} 

Any classical strategy for generating a given tripartite distribution in the triangle scenario can geometrically be represented by a cube with labels on it sides.
Consider a unit cube in three dimensions. We label three mutually orthogonal edges of the cube by the three sources $\alpha, \beta, \gamma$ and label $A$, $B$, $C$ the faces respectively orthogonal to the edges $\alpha, \beta, \gamma$. Recall that we assume that the sources $\alpha, \beta$ and $\gamma$ take values in $[0,1]$. Thus any values of $\beta, \gamma\in [0,1]$ correspond to a point on face $A$, and to an answer of Alice  when she receives $(\beta,\gamma)$ from the sources, i.e., $a(\beta,\gamma)$. That is, points of face $A$ are labeled by $a$'s. On the other hand, as we mentioned before, for each source the interval $[0, 1]$ is divided in a finite number of subintervals, and the parties are ignorant of the exact value of the source, but its subinterval index. Therefore, face $A$ is indeed partitioned in some aligned rectangles which are labeled by $a$'s.
The same applies to faces $B$ and $C$. See Fig.~\ref{strat_triangle_appendix}.

\subsection{Basic properties of $\NL$ in the triangle scenario}

In the following, we recall and give basic properties of the set of correlation in $\NL$ for the triangle scenario. We choose to illustrate some proofs with the cube representation of strategies discussed above. 
For simplicity of presentation, we limit ourselves to the case where the outputs are all binary; generalization of these proofs to larger output alphabet sizes is straightforward. 

\begin{proposition}
Let $\mathcal N$ denote the triangle network. Then the followings hold:
\begin{enumerate}[label=(\roman*)] 
\item  The GHZ distribution and the W distribution (in which exactly one of the parties, chosen uniformly at random, outputs $1$ and the others output $0$) are not in $\NL$ (see also~\cite{wolfe}).

\item $\NL$, as well as $\NQ$, $\NB$, are contractible (even though they are not star convex), hence do not contain holes.

\item $\NL$ contains an open ball in the probability simplex around $P_u$ the maximally mixed distribution (see also \cite{evans}). Thus the dimension of $\NL$ equals the dimension of the probability simplex, which is $7$ when outputs are all binary. The same holds for $\NQ$ and $\NB$ as they contain $\NL$.

\item The Finner inequality \eqref{Finner_triangle} is valid for any distribution in $\NL$.
\end{enumerate}
\end{proposition}

\begin{proof}
($i$) can be proven based on the cube representation of strategies.
The proof for the W distribution is given in Fig.~\ref{NoStratW}; the proof for GHZ is left for the reader.
\begin{figure}[H]
\begin{center}
\includegraphics[scale=0.22]{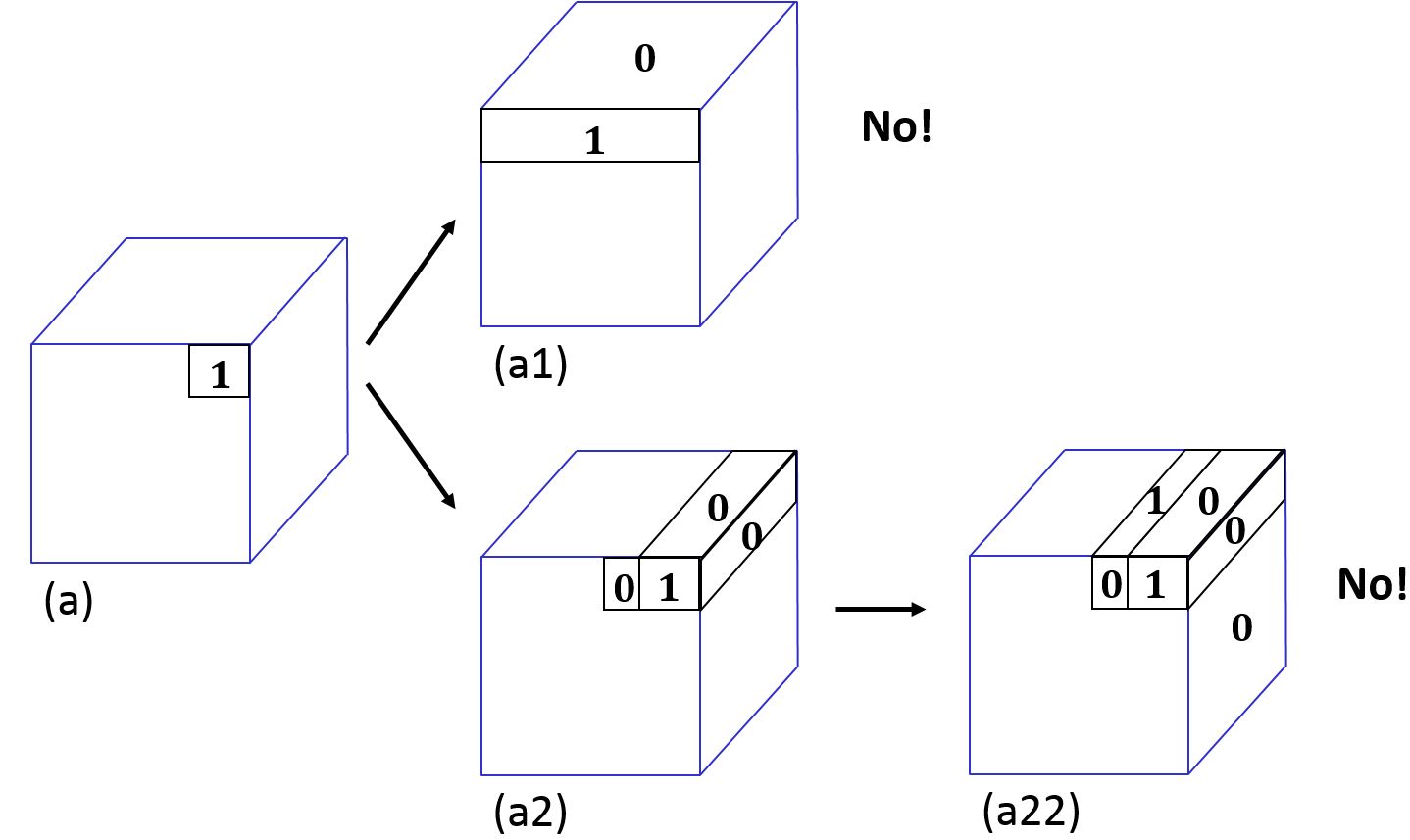}
\caption{Suppose that a cube gives the W distribution. As Alice sometimes answer 1, by relabeling the hidden variables, we can suppose that there is a 1 at the location given in $(a)$. Then, there are two possibilities: either there is no 0 on the left of that 1 (case $(a1)$), or there is one (case $(a2)$). The first case is not possible: as Alice and Charlie never answer 1 together, there must be 0 everywhere on Charlie's face, i.e., Charlie always answers 0.
In the second one, as when Alice says 1 Bob and Charlie cannot say 1, we end up with the cube $(a2)$. As when Alice and Bob both say 0, Charlie must say 1, we obtain the cube $(a22)$. However, in that cube Bob should not say 1 when Charlie already says it; he must answer 0 all the time, which is absurd.}\label{NoStratW}
\end{center}
\end{figure}

($ii$) Recall that a set $\SSS$ is said to be contractible if it can be continuously shrunk to a point within $\SSS$.
More precisely, there is a \emph{continuous} map $\Phi: (t, P)\in [0,1]\times \SSS \mapsto \Phi_t(P)\in\SSS$  such that $\Phi_1(P)=P$ and $\Phi_0(P)=P_u$ for some fixed $P_u$.
For the set $\NL$ such a map $\Phi_t(P)$ for an arbitrary $P\in \NL$ is constructed as follows, and can similarly be defined for $\NQ$ and $\NB$. 

In the cubic representation of strategies, since the parties can also have access to local randomness independent of common sources, we may add a question mark symbol `$?$' telling the party to choose her output uniformly at random. Thus, the maximally mixed (uniform) distribution $P_u$ corresponds to a cube all of whose three orthogonal faces are labeled by the question mark. Now consider a strategy for generating a distribution $P$ and construct a cube whose corner $t\times t\times t$ sub-cube, for some $0\leq t\leq 1$, is filled according to the \emph{renormalized} strategy for $P$, and the rest of it is filled by the question mark (see Fig.~\ref{contractile}). Call the resulting distribution $P_t=\Phi_P(t)$. When the outputs are binary, a simple computation verify that 
\begin{align}\label{eq:P_t-abc}
P_t(abc) =& \frac{1}{8}\Big(1-3t^2+2t^3 + 8 t^3 P(abc) \nonumber\\
&\qquad + 2t^2(1-t)  \big(P_A(a)+P_B(b)+P_C(c)\big)   \Big).
\end{align}
Clearly, $P_t=\Phi_P(t)$ is continuous in $(P,t)$.

\begin{figure}[H]
\begin{center}
\includegraphics[scale=0.6]{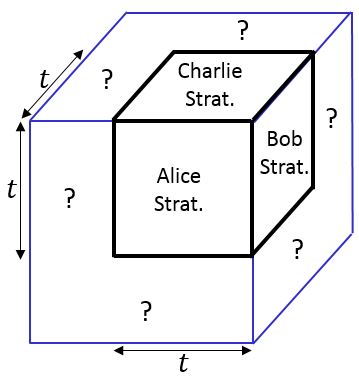}
\caption{ In the cube representation of strategies, fill the upper corner cube of size $t\times t\times t$ according to a strategy for $P$, and with full random choices elsewhere. The resulting distribution is denoted by $P_t=\Phi_P(t)$  and is given by~\eqref{eq:P_t-abc}.
}
\end{center}
\label{contractile}
\end{figure}

($iii$) Let $P^{(j)}$, $j=1, \dots, k$, be some distributions in $\NL$. Pick arbitrary $\epsilon_j\geq 0$ with $\sum_j \epsilon_j\leq 1$ and on each of the three orthogonal sides of the cube pick disjoint intervals of sizes $\epsilon_j$ for any $j$. Then in the cube one finds $k$ (disjoint) sub-cubes of sizes $\epsilon_j\times \epsilon_j\times \epsilon_j$ for any $j$ (see Fig.~\ref{ballP0strategies}). Now similar to the proof of part ($ii$) fill the $j$-th sub-cube according to the scaled strategy associated to $P^{(j)}$. This gives a distribution $Q\in \NL$ which can be derived following similar computation as that of~\eqref{eq:P_t-abc}:
\begin{align*}
Q(abc) = &\frac{1}{8}\Big(    1+ \sum_j \Big[-3\epsilon_j^2+2\epsilon_j^3 + 8 \epsilon_j^3 P^{(j)}(abc) \nonumber\\
&\quad + 2\epsilon_j^2(1-\epsilon_j)  \big(P^{(j)}_A(a)+P^{(j)}_B(b)+P^{(j)}_C(c)\big)\Big]  \Big),
\end{align*}
where 1 has to be interpreted as a vector full of ones (similarly for $\frac{1}{8}$ and $\frac{1}{2}$ in the following).

To continue the proof it is instructive two write down the distributions $P^{(j)}$ as 
$$P^{(j)} = \frac{1}{8}+ R^{(j)},$$
with marginals $P^{(j)}_A = \frac{1}{2}+ R_A^{(j)}$ etc. Then letting $Q=\frac{1}{8} + S$ the above equation can be rewritten as 
\begin{align}\label{eq:S-abc}
S(abc) = &\frac{1}{8} \sum_j \Gamma_{P^{(j)}, \epsilon_j}(abc),
\end{align}
where 
\begin{align*}
\Gamma_{P^{(j)}, \epsilon_j}(abc)=&8 \epsilon_j^3 R^{(j)}(abc)\\ 
& + 2\epsilon_j^2(1-\epsilon_j)  \big(R^{(j)}_A(a)+R^{(j)}_B(b)+R^{(j)}_C(c)\big).
\end{align*}

Now to finish the proof we need to show that any $S$ satisfying $\sum_{a, b, c} S(abc)=0$  and with sufficiently small coordinates can be written as~\eqref{eq:S-abc} for some $P^{(j)}\in \NL$ and some choices of $\epsilon_j$'s.

Let $P^{(xyz)}$, for $(x, y, z)\in \{0,1 , ?\}^3$ be the distribution coming from the cube whose Alice's face is labeled $x$, Bob's face is labeled $y$, and Charlie's face is labeled $z$. For instance we have
$$P^{(0??)} (abc) = \frac{1}{4}\delta_{a=0},$$
with $R^{(0??)}_A(a) = \frac 12 (-1)^a$ and $R^{(0??)}_B=R^{(0??)}_C=0$.
Then we have 
$$\Gamma_{P^{(0??)}, \epsilon}(abc) = \epsilon^2 (-1)^a,$$
and also $\Gamma_{P^{(1??)}, \epsilon}(abc) = - \Gamma_{P^{(0??)}, \epsilon}(abc)$.
We similarly can compute
$$\Gamma_{P^{(00?)}, \epsilon}(abc) = \epsilon^2\big( (-1)^a+(-1)^b \big) + \epsilon^3 (-1)^{a+b},$$
and
\begin{align*}
\Gamma_{P^{(000)}, \epsilon}(abc) = &\epsilon^2\big( (-1)^a+(-1)^b +(-1)^c\big) \\
&+ \epsilon^3 \big( (-1)^{a+b} +(-1)^{a+c} +(-1)^{b+c}  \big)\\
& + \epsilon^3(-1)^{a+b+c}.
\end{align*}
Comparing the above equations, we find that by considering the summations of these $\Gamma$ terms for different choices of $(x, y, z)\in \{0, 1, ?\}^3$, we can write the functions 
\begin{align*}
&\pm\delta (-1)^{a}, \pm\delta (-1)^{b}, \pm\delta (-1)^{ c},\\
&\pm\delta (-1)^{ a +b} , \pm\delta (-1)^{a+c}, \pm\delta (-1)^{a+ c}\\
&\pm\delta (-1)^{+ a+ b+ c},
\end{align*}
in the form of~\eqref{eq:S-abc} when $\delta$ is sufficiently small. Observing that these functions, which also include their negations, form a basis for the space of functions $S$ with $\sum_{abc} S(abc)=0$, the proof is concluded. 

Reference \cite{evans} is related to the same question and exploits a totally different framework. 

\begin{figure}[H]
\begin{center}
\includegraphics[scale=0.45]{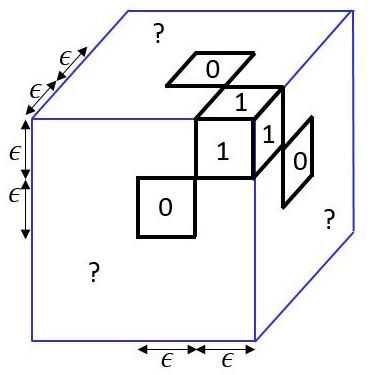}
\caption{ Suppose that two disjoint sub-cubes of size $\epsilon\times \epsilon\times \epsilon$ (with disjoint projections on the three orthogonal directions) are filled with constant $0$ and constant $1$, and the rest of the cube with totally random choices. The resulting tripartite distribution would be equal to $P=P_u+\frac{\epsilon^3}{4}V$ where $P_u$ is the maximally mixed (uniform) distribution and $V$ is given by $V(000)=V(111)=3$ and $V(abc)=-1$ if $(a, b,c)\in \{0,1\}^3\setminus\{000, 111\}$. 
}
\label{ballP0strategies}
\end{center}
\end{figure}

($i\nu$) This is a direct consequence of the Loomis-Whitney inequality, asserting that in $\mathbb R^3$, the square of the volume of any measurable subset is bounded by the product of the areas of its projections in the three orthogonal directions.

\end{proof}

\section{The Finner inequality in term of the Hypercontractivity Ribbon}\label{app:HR}

The Hypercontractivity Ribbon (HR) is a measure of correlation that can be defined in terms of parameters for which the Finner inequality is satisfied. In the tripartite case on which we focus, the HR of a distribution $P_{ABC}$ is the set of $(\alpha, \beta, \gamma)\in [0,1]^3$ for which
\begin{equation}\label{hypercontractivity_ribbon_2}
\mean{f_Ag_Bh_C}\leq  \|f_A\|_{1/\alpha}\cdot \|g_B\|_{1/\beta}\cdot \|h_C\|_{1/\gamma}, \end{equation}
for all choices of functions $f_A, g_B$ and $h_C$. We denote the HR of $P_{ABC}$ by $\fR(A, B, C)$. An important property of HR is that it expands under local post-processing. That is, if $A', B', C'$ are obtained by local post-processing of $A, B, C$ respectively, then we have 
\begin{align}\label{eq:HR-monotone}
\fR(A, B, C)\subseteq \fR(A', B', C').
\end{align}
More interesting is the tensorization property of HR saying that $\fR(A^n, B^n, C^n) = \fR(A, B, C)$ where the former is computed with respect to the iid distribution $P_{ABC}^{\otimes n}$. See~\cite{beigi} and references therein for more details.

Finner's inequality can be stated in terms of HR. In the triangle scenario, for instance, Finner's inequality says that for every $P_{ABC}\in \NL$ we have $(1/2, 1/2, 1/2)\in \fR(A, B, C)$. In general, Finner's inequality says that any fractional independent set of a network belongs to the HR of any distribution achievable in that network.

A crucial property of $\fR(A,B,C)$ is that it can be expressed in terms of the mutual information function as follows. $\fR(A, B, C)$ consists of the set of non-negative triples $(\alpha, \beta, \gamma)$ such that for any \emph{auxiliary} random variable $U$ given by $P_{U|ABC}$ we have   
\begin{equation}\label{hypercontractivity_ribbon_1}
I(U;ABC)\geq \alpha I(U;A) + \beta I(U;B) + \gamma I(U;C) 
\end{equation}
From this characterization of HR it is clear that $\fR(A, B, C)$ is a convex set. This is a property that will be used in the proof of Theorem~\ref{thm:quantum-finner}.

\section{The Finner inequality holds in $\NQ$}
\label{finner_quantum}

We now present a proof of Theorem~\ref{thm:quantum-finner}, that the Finner inequality holds for any $P\in \NQ$ when the sources in the network $\mathcal N$ are all bipartite. 

First of all, as mentioned in Appendix~\ref{app:HR}, given a distribution $P$, the set of weights $\boldsymbol\eta = (\eta_1, \dots, \eta_n)$ for which the Finner inequality~\eqref{NQFinner} holds for all choices of $f_j$'s, is a convex set. That is, if~\eqref{NQFinner} holds for $\boldsymbol \eta$ and $\boldsymbol{\eta'}$, then it holds for any convex combination of them. Therefore, in order to show that the Finner inequality is satisfied for all weights $\boldsymbol \eta$ that form a fractional independent set (which itself is a convex set), it suffices to prove it for the extreme points of the set of fractional independent set. That is, in the proof we may assume that $\boldsymbol \eta =(\eta_1, \dots, \eta_n)$ is an extreme point of the set of fractional independent sets.

Second, we use the assumption that all sources in the network $\mathcal N$ are bipartite. It is well-known that in any graph the extreme points of the set of fractional independent sets are half-integers~\cite[Theorem 64.7]{Schrijver}. In other words, if $\boldsymbol \eta =(\eta_1, \dots, \eta_n)$ is an extreme fractional independent set, for all $j$ we have $\eta_j\in \{0, 1/2, 1\}$.
On the other hand, if one of $\eta_j$'s, say $\eta_n$, equals $0$, then we have $$\|f_n\|_{1/\eta_n} = \|f_n\|_\infty = \max_{a_n: P_{A_n}(a_n)\neq 0} f(a_n),$$
and
$$\mathbb E\Big[\prod f_j\Big] \leq \mathbb E\bigg[\,\prod_{j=1}^{n-1} f_j\,\bigg] \cdot \|f_n\|_\infty.$$
This means that if~\eqref{NQFinner} holds ignoring the $n$-th party (for the marginal distribution $P_{A_1, \dots, A_{n-1}}$), it also holds including her and putting $\eta_n=0$. We conclude that we may restrict ourselves to weights $\boldsymbol \eta =(\eta_1, \dots, \eta_n)$ such that $\eta_j\in \{1/2, 1\}$ for all $j$. Even more, for such weights if there is $j$, say $j=n$, with $\eta_j=1$, then the $n$-th party cannot share any source with others. This is because if $A_{j'}$ shares a source with $A_n$ then we must have $\eta_n+\eta_{j'}\leq 1$ that is a contradiction since $\eta_{j'}$ is assumed to be in $\{1/2, 1\}$. This means such a party $A_n$ with $\eta_n=1$ is isolated and shares nothing with others. In this case we have
$$\mathbb E\Big[\prod f_j\Big] =  \mathbb E\bigg[\,\prod_{j=1}^{n-1} f_j\,\bigg] \cdot \mathbb E[f_n],$$
and of course $\|f_n\|_{1/\eta_n} = \|f_n\|_1= \mathbb E[f_n]$. As a result, parties whose weights are equal to $1$ can be ignored. Putting all these together we may assume that all the weights are equal to $\eta_j=1/2$ and we need to prove
\begin{align}\label{eq:Q-Finner-2}
\mean{\prod_j f_j} \leq \prod_j \norm{f_j}_{2},
\end{align}

Third, recall that $f_j$ is an arbitrary function applied on the output of the $j$-th party $A_j$. Composing the measurement operators of $A_j$ with this classical post-processing, we may assume that $f_j$ is the outcome of some quantum observable $X_j$ that the $j$-th party applies on quantum systems in her hand. Indeed, we may put
$$X_j = \sum_{a_j} f_j(a_j) M_{a_j}^{(j)},$$
where $\big\{M_{a_j}^{(j)}\big\}$ is the projective measurement applied by $A_j$.
Thus we have
\begin{align}\label{eq:mean-f-j-X}
\mathbb E\Big[ \prod_j f_j \Big] = \text{Tr} \Big[ \rho\cdot \bigotimes_j X_j  \Big],
\end{align}
where as before $\rho=\bigotimes_i \rho_i$ and $\rho_i$ is the pure state associated to the $i$-th source. 


Recall that for each source $i$, $\rho_i$ is a pure bipartite state, so we may consider its Schmidt decomposition. For simplicity of notation we may assume that the dimension of all subsystems in $\rho_i$'s are equal (by taking the maximum of all these local dimensions). Moreover, by applying appropriate local rotations we may assume that the Schmidt basis of all $\rho_i$'s are the same.  Thus we may write $\rho_i = \ket{\psi_i}\bra{\psi_i}$ with
$$\ket{\psi_i} = \sum_{\ell=1}^d  \lambda_\ell^{(i)} \ket{\ell}\otimes \ket{\ell},$$
where $\lambda_\ell^{(i)}\geq 0$ denote Schmidt coefficients of $\rho_i$.
Then~\eqref{eq:mean-f-j-X} reduces to
\begin{equation*}
\mathbb E\Big[ \prod_j f_j \Big] = \sum_{\stackrel{\ell_1, \dots, \ell_m} {\ell'_1, \dots, \ell'_m}} \prod_i  \lambda^{(i)}_{\ell_i} \lambda^{(i)}_{\ell'_i} \cdot \prod_j  \text{Tr}\Big[ X_j\cdot \bigotimes_{i: i\to j} \ket{\ell_i}\bra{\ell'_i}   \Big].
\end{equation*}

For any $j$ define $g_j: \prod_{i: i\to j} \{1, \dots, d\}^2\to \mathbb R$ by
$$g_j\Big(  (\ell_i, \ell'_i )_{i: i\to j} \Big) = \prod_{i:i\to j} \sqrt{\lambda^{(i)}_{\ell_i} \lambda^{(i)}_{\ell'_i}} \cdot \text{Tr} \Big[ X_j\cdot  \bigotimes_{i: i\to j} \ket{\ell_i}\bra{\ell'_i}\Big].$$
Also, let $R_i$ be the uniform random variable taking values in $\{1, \dots, d\}^2$, i.e., $R_i$ equals $(\ell, \ell')$ with probability $d^{-2}$.  Then the previous equation can be written as
\begin{equation*}
\mathbb E\Big[ \prod_j f_j \Big] = d^{2m} \mathbb E\Big[    \prod_j g_j\big(  (R_i)_{i\to j} \big)     \Big]
\end{equation*}
Now we may think of $R_i$'s as sources of randomnesses that are shared to the parties who apply local functions $g_j$ on them. Then by the (classical) Finner inequality we have
$$\mathbb E\Big[ \prod_j f_j \Big]\leq d^{2m} \prod_j \|g_j\|_{2}.$$
Let us compute the factors on the right hand side:
\begin{align*}
 &d^{2\cdot|\{i:\, i\to j\}|}\cdot\|g_j\|^2_{2} =  d^{2\cdot|\{i:\, i\to j\}|}\cdot\mathbb E\big[ g_j^2 \big]\\
& =  \sum_{(\ell_i, \ell'_i)_{i: i\to j}}  \prod_{i:i\to j} \lambda^{(i)}_{\ell_i} \lambda^{(i)}_{\ell'_i} \cdot \text{Tr} \Big[ X_j\cdot  \bigotimes_{i: i\to j} \ket{\ell_i}\bra{\ell'_i}\Big]^2\\
& = \text{Tr} \big[ \sqrt{\sigma_j} X_j\sqrt{\sigma_j} X_j\big],
\end{align*}
where 
$$\sigma_j= \sum_{(\ell_i)_{i:i\to j}} \Big(\prod_{i:i\to j}\lambda^{(i)}_{\ell_i}\Big)^2 \bigotimes_{i:i\to j} \ket{\ell_i}\bra{\ell_i}.$$
We continue 
\begin{align*}
d^{2\cdot|\{i:\, i\to j\}|}\cdot\|g_j\|^2_{2} & =  \text{Tr} \Big[ \big(\sigma_j^{1/4} X_j\sigma_j^{1/4}\big)^2\Big]\\
& \leq \text{Tr} \Big[  \sigma_j^{1/2}X_j^2 \sigma_j^{1/2}   \Big]\\
&= \text{Tr} \Big[  \sigma_j X_j^2    \Big]\\
& = \|f_j\|_2^2,
\end{align*}
where the inequality follows from the Araki-Lieb-Thirring inequality, and the last equality is verified by an easy computation. We conclude that
\begin{align*}
\mathbb E\Big[ \prod_j f_j \Big]&\leq d^{2m} \prod_j \|g_j\|_{2}\\ 
&= \prod_j d^{|\{i:\, i\to j\}|} \cdot \|g_j\|_2\\
&\leq \prod_j \|f_j\|_2,
\end{align*}
where the equality follows from the fact that the sources are bipartite and for each $i$ there are exactly two $j$'s for which $i\to j$.

\section{Finner inequality holds for triangle scenario in the Boxworld}\label{finner_boxworld_proof}

In this section we show that the Finner inequality is satisfied for the triangle scenario in the Boxworld, where bipartite no-signaling boxes are wired by the parties to produce an output.  As explained in Appendix \ref{finner_quantum}, we can restrict ourselves to coefficients 1/2.


\begin{theorem} [Finner inequality in Boxworld]\label{finner_bw}
Letting $\mathcal N$ be the triangle network, for any $P_{ABC}\in \NB$ we have
$$\mean{f_Ag_Bh_C}\leq  \|f_A\|_{2}\cdot \|g_B\|_{2}\cdot \|h_C\|_{2}.$$
\end{theorem}

Before getting into the details of the proof, let us briefly explain the proof ideas. First, as mentioned in Appendix~\ref{app:HR}, the above theorem says that for any $P_{ABC}\in \NB$ we have $(1/2, 1/2, 1/2)\in \fR(A, B, C)$. On the other hand, since HR is monotone under local post-processing (equation~\eqref{eq:HR-monotone}), it suffices to prove that $(1/2, 1/2, 1/2)\in \fR(T, R, S)$ where $T, R$ and $S$ denote all information available to Alice, Bob and Charlie respectively, at the end of the wirings. This is because, $A, B, C$ are functions of $T, R, S$ respectively.  
Next we can use the second equivalent characterization of HR, and in order to prove $(1/2, 1/2, 1/2)\in \fR(T, R, S)$ show that  
\begin{align*}
    \chi = I(U; TRS) - \frac{1}{2}I(U;T) - \frac{1}{2}I(U;R) - \frac{1}{2}I(U;S)\ge 0. 
\end{align*}
The proof of this inequality is based on the chain rule of mutual information. In the wiring, each party uses her boxed in hand one by one: for each time-step new information is added to her transcript. Therefore, we expand each mutual information term in the above equation as a summation over time-steps using the chain rule. 
We further expand each time step into the choice of box and input, and the creation of the output. 
As a result we obtain $\chi = \chi_I + \chi_O$, where $\chi_I$ (respectively, $\chi_O$) is a summation of terms corresponding to all the choices of boxes and their inputs (respectively, outputs) by the parties.

As the choice of boxes and their inputs are done locally and independently of the sources, $\chi_I$ can easily be bounded. As the order in which Alice, Bob and Charlie choose there boxes is a priori not the same, bounding $\chi_O$ is more tricky. One has to first reorder the summation in $\chi_O$ to put terms associated to a given box together, and then use properties of mutual information to bound $\chi_O$. We use the no-signaling assumption in this last step.

This proof follows the ideas introduced in \cite{beigi}, in which the author prove that PR boxes cannot be purified using wirings.

\subsection{Notations}

Here we introduce the notations we need to prove Theorem~\ref{finner_bw}. We define notations for Alice and then recap them for Bob and Charlie.

Let $J_{AB}$ (respectively $J_{AC}$) be the set of index of all boxes shared between Alice and Bob (respectively Alice and Charlie). Let $N_A$ be the number of boxes available to Alice, i.e., $N_A = |J_{AB}|+|J_{AC}|$. 

For $j \in J_{AB}$, let $X_j$, $Y_j$ (respectively $A_j$, $B_j$) be the inputs (respectively, the outputs) of the box $j$. This box is described by a no-signaling conditional distribution $P(A_j B_j|X_j Y_j)$.

In each time-step Alice chooses which box to use next and its input as a random function of whatever she has so far, i.e. previous choices of boxes, their inputs and their outputs. We denote by $\Pi_j$ the time-step at which Alice uses box $j$. Thus $\Pi$ is a permutation of the boxes available to Alice. We denote the inverse of this distribution by $\tilde{\Pi}$. That is, $\tilde{\Pi}_i$ is the index of the box used by Alice at time-step $i$.

Let $T_j$ be the \emph {transcript} of Alice before using the $j$-th box, i.e., whatever she has seen before using the $j$-th box. We also denote Alice's \emph{extended transcript} by ${T}_j^e$ that is $T_j$ together with $\Pi_j$ and $X_j$: 
$$T_j^e=( T_j,\Pi_j,X_j).$$ 
We also use the notations
$\tilde{X}_i=X_{\tilde{\Pi}_i}$, $\tilde{A}_i=A_{\tilde{\Pi}_i}$, $\tilde{T}_i=T_{\tilde{\Pi}_i}$ and $\tilde{T}^e_i=T^e_{\tilde{\Pi}_i}$. Observe that, for instance, $\tilde X_i$ is the box that Alice uses in time-step $i$. With these notations we have
$$T_j := \big ( \tilde{\Pi}_1, \dots, \tilde{\Pi}_{\Pi_{j-1}}, \tilde{X}_1, \dots, \tilde{X}_{\Pi_{j-1}}, \tilde{A}_1, \dots, \tilde{A}_{\Pi_{j-1}} \big ).$$

Here is a summary of notations for later use:
\begin{itemize}
    \item $\Pi_j$ : Alice uses the $j$-th box in her $\Pi_j$-th action. 
    \item $\tilde{\Pi}_i$: Alice uses the $\tilde{\Pi}_i$-th box in her $i$-th action.
    \item $X_j$ : Alice's input of the $j$-th box.
    \item $\tilde{X}_i$ : Alice's input in her $i$-th action.    
    \item $A_j$ : Alice's output of the $j$-th box.  
    \item $\tilde{A}_i$ : Alice's output in her $i$-th action.
    \item $T_j$ : Alice's transcript before using the $j$-th box.
    \item $\tilde{T}_i$ :Alice's transcript before her $i$-th action. 
    \item $T_j^e=\{ T_j,\Pi_j,X_j\}$.
    \item $\tilde{T}_i^e=\{ \tilde{T}_i,\tilde{\Pi}_i,\tilde{X}_i\}$. 
\end{itemize}

We use superscript $N_A$ to denote the full set of variables at the end, e.g., $A^{N_A}= (A_1, \ldots, A_{N_A})$. At the end, Alice determines her final output by applying a stochastic map on all information available to her, i.e., on $\Pi^{N_A}, X^{N_A}, A^{N_A}$ which we denote by
\begin{align}\label{eq:T-transcript}
T=\big(\Pi^{N_A}, X^{N_A}, A^{N_A}\big).
\end{align}

At the end Alice, Bob, and Charlie apply their stochastic maps to determine their final output.

The corresponding variables in the above list for Bob are  $\Omega_j, \tilde{\Omega}_i, Y_j, \tilde{Y}_i, B_j, \tilde{B}_i, S_j, \tilde{S}_i, S_j^e, \tilde S_i^e$ respectively. 
The corresponding variables in the above list for  Charlie are $\Gamma_j, \tilde{\Gamma}_i, Z_j, \tilde{Z}_i, C_j, \tilde{C}_i, R_j, \tilde{R}_i, R_j^e, \tilde R_i^e$ respectively

\subsection{Auxiliary lemmas}

In this section, we introduce some lemmas deduced from the no-signaling condition.

\begin{lemma}\label{lemma:input_part}
For every $1\leq i\leq N_B$ we have
$$I(\tilde{Y}_i \tilde{\Omega}_i; T | \tilde{S}_i  ) = 0,$$
and for every $1\leq i\leq N_C$ we have
$$I(\tilde{Z}_i \tilde{\Gamma}_i; T S | \tilde{R}_i  ) = 0.$$
\end{lemma}

\begin{proof}
These expressions are simple consequences of the fact that each party at time-step $i$ chooses a box and its input locally as a (stochastic) function of the transcript at step $i$, and other parties cannot signal using the boxes. 

\end{proof}

\begin{lemma}\label{B_given_A}
For boxes available to Bob we have
\begin{enumerate}
\item[(i)] For $j \in J_{AB}$:
$$ I(B_j; T|T_j^e A_j S_j^e) = 0$$
\item[(ii)] For $j \in J_{BC}$:
$$I(B_j;T |S_j^e) = 0$$
\end{enumerate}
\end{lemma}

\begin{proof}
$(i)$ states the independence of Bob's output of the $j$-th box and the future information in Alice's side, given all the information of Alice and Bob, except $B_j$, up to just after using this box.
To prove this, it is enough to show that $H(B_j|T_j^e A_j S_j^e)=H(B_j|T S_j^e)$, for which we compute
\begin{align*}
H(B_j|T S_j^e) & = H(A_jB_j|T\setminus\{A_j\} S_j^e) - H(A_j| T\setminus\{A_j\} S_j^e) \\
& = H(A_jB_j|T_j^e S_j^e) - H(A_j| T_j^e S_j^e)\\
& = H(B_j| T_j^e A_j S_j^e).
\end{align*}
Here the second line follows from the fact that $A_j$ and $B_j$ are determined independently of the other variables once the inputs of the $j$-th box are fixed.\\

$(ii)$ states that given the input of Bob for a box shared between Bob and Charlie, Bob's output is independent of Alice's transcript.
To prove this we compute
\begin{align*}
I(T;B_j|S_j^e) & \leq  I(TR_j^e; B_j| S_j^e) \\
& = I(T;B_j|S_j^e R_j^e)\\
& \leq I(T;B_j C_j|S_j^e R_j^e).
\end{align*}
Here the inequalities follow from the data processing inequality, and the equality follows from the chain rule and $I(B_j;R_j^e|S_j^e) = 0$, the no-signaling condition. Then the desired result follows once we note that $I(T;B_j C_j|S_j^e R_j^e)=0$ since the output of the $j$-th box are determined independently of other variables once its inputs are fixed. 

\end{proof}

The following lemma presents similar statements as above for boxes available for Charlie. We skip its proof as it follows from similar ideas as above. 

\begin{lemma}\label{C_given_AB}
For boxes available to Charlie we have
\begin{enumerate}
\item[(i)] For $j \in J_{BC}$:
$$I(C_j;T S|S_j^e B_j R_j^e) = 0$$
\item[(ii)] For $j \in J_{AC}$:
$$I(C_j;S T|T_j^e A_j R_j^e) = 0$$
\end{enumerate}
\end{lemma}

In the following for four random variables $X, Y, Z, W$ we use the notation
\begin{align*}
I(X; Y; Z|W) =&H(X|W)+H(Y|W)+H(Z|W)  \\
&- H(XY|W)-H(XZ|W)-H(YZ|W)\\
&+ H(XYZ|W).
\end{align*}
We will frequently use the following expression for $I(X; Y; Z|W)$ which can easily be verified:
\begin{align}\label{important0}
I(X; Y; Z) = I(X; Y|W) - I(X; Y|WZ).
\end{align}
We indeed use the symmetry in the definition of $I(X; Y; Z|W)$ which gives
\begin{align}\label{important}
I(X; Y|W) - I(X; Y|WZ) = I(Y; Z|W) - I(Y; Z| WX).
\end{align}

\subsection{Proof of Theorem \ref{finner_bw}} 

As mentioned before, we need to show that $(1/2, 1/2, 1/2)\in \fR(A, B, C)$. Moreover, since HR satisfies the monotonicity property~\eqref{eq:HR-monotone}, and $A, B, C$ are determined by post-processing of $T, S, R$ respectively, it suffices to prove that $(1/2, 1/2, 1/2)\in \fR(T, S, R)$. That is, we need to show that for any $P_{U|TSR}$ we have
\begin{align*}
    \chi = I(U; TRS) - \frac{1}{2}I(U;T) - \frac{1}{2}I(U;R) - \frac{1}{2}I(U;S)\ge 0. 
\end{align*}

We first write
\begin{align}\label{eq:chain-order}
I(U; TRS)=I(U;T) + I(U;S|T) +I(U;R|T S).
\end{align}
Then noting that, say, $T$ itself consists of several random variables as in~\eqref{eq:T-transcript}, we apply chain rule once again to each of the above terms.
This decomposes $\chi$ into two terms $\chi = \chi_I + \chi_O$ associated to the \emph{input parts} and the \emph{output parts} given by
\begin{align*}
\chi_I  = &\sum_{i=1}^{N_A} \Big[ I(U; \tilde{X}_i \tilde{\Pi}_i|\tilde{T}_i) - \frac{1}{2} I(U; \tilde{X}_i \tilde{\Pi}_i | \tilde{T}_i) \Big] \\
&+ \sum_{i=1}^{N_B} \Big[ I(U;\tilde{Y}_i\tilde{\Omega}_i|T \tilde{S}_i)  - \frac{1}{2} I(U; \tilde{Y}_i \tilde{\Omega}_i | \tilde{S}_i) \Big]\\
&+\sum_{i=1}^{N_C} \Big[ I(U;\tilde{Z}_i \tilde{\Gamma}_i| T S \tilde{R}_i)-\frac{1}{2}I(U;\tilde{Z}_i\tilde{\Gamma}_i|\tilde{R}_i)\Big]
\end{align*}
and
\begin{align*}
\chi_O = &\sum_{i=1}^{N_A} \Big[ I(U; \tilde{A}_i |\tilde{T}_i^e) - \frac{1}{2} I(U; \tilde{A}_i | \tilde{T}_i^e)  \Big] \\
&+ \sum_{i=1}^{N_B} \Big[ I(U;\tilde{B}_i|T \tilde{S}_i^e)  - \frac{1}{2} I(U; \tilde{B}_i | \tilde{S}_i^e) \Big]\\ 
&+ \sum_{i=1}^{N_C} \Big[ I(U;\tilde{C}_i| T S \tilde{R}_i^e)-\frac{1}{2}I(U;\tilde{C}_i|\tilde{R}_i^e)\Big]
\end{align*}
We will show separately that both $\chi_I$ and $\chi_O$ are non-negative.

Let us first start with $\chi_I\geq 0$ that is easy (as each party chooses its box and input independently, a stronger inequality holds with the terms $1/2$ replaced by $1$). The first summand in $\chi_I$ is non-negative since $I(U; \tilde{X}_i \tilde{\Pi}_i|\tilde{T}_i) - \frac{1}{2} I(U; \tilde{X}_i \tilde{\Pi}_i | \tilde{T}_i)= \frac{1}{2} I(U; \tilde{X}_i \tilde{\Pi}_i | \tilde{T}_i)\geq 0$. For the second summand we compute
\begin{align*}
&I(U;\tilde{Y}_i\tilde{\Omega}_i|T \tilde{S}_i)  - \frac{1}{2} I(U; \tilde{Y}_i \tilde{\Omega}_i | \tilde{S}_i) \\
 &\qquad\geq I(U;\tilde{Y}_i\tilde{\Omega}_i|T \tilde{S}_i)  - I(U; \tilde{Y}_i \tilde{\Omega}_i | \tilde{S}_i)\\
 &\qquad = - I(U; \tilde{Y}_i\tilde{\Omega}_i; T| \tilde S_i)\\
 &\qquad = I(T ;\tilde{Y}_i\tilde{\Omega}_i| U\tilde{S}_i)  - I(T ;\tilde{Y}_i\tilde{\Omega}_i| \tilde{S}_i)\\
 &\qquad = I(T ;\tilde{Y}_i\tilde{\Omega}_i| U\tilde{S}_i)\\
 &\qquad \geq 0,
\end{align*}
where the third equality follows from Lemma~\ref{lemma:input_part}. The proof that the third summand is non-negative is similar. Therefore, $\chi_I \geq 0$.

We now show that $\chi_O \geq 0$. Observe that by definitions $\tilde A_i= A_{\tilde \Pi_i}$ etc. Then splitting the summands in $\chi_O$ in terms of boxes shared between different pairs of parties, we obtain:
\begin{align*}
\chi_O =& \sum_{i:\,\tilde{\Pi}_i \in J_{AB}} \Big[ I(U; A_{\tilde{\Pi}_i}|T_{\tilde{\Pi}_i}^e) - \frac{1}{2} I(U; A_{\tilde{\Pi}_i} | T_{\tilde{\Pi}_i}^e)  \Big] \\
&+\sum_{i:\,\tilde{\Pi}_i \in J_{AC}} \Big[ I(U; A_{\tilde{\Pi}_i}|T_{\tilde{\Pi}_i}^e) - \frac{1}{2} I(U; A_{\tilde{\Pi}_i} | T_{\tilde{\Pi}_i}^e)  \Big] \\
&+ \sum_{i:\,\tilde{\Omega}_i \in J_{AB}} \Big[ I(U; B_{\tilde{\Omega}_i}|T S_{\tilde{\Omega}_i}^e) - \frac{1}{2} I(U; B_{\tilde{\Omega}_i} | S_{\tilde{\Omega}_i}^e)  \Big] \\
&+ \sum_{i:\,\tilde{\Omega}_i \in J_{BC}} \Big[ I(U; B_{\tilde{\Omega}_i}|T S_{\tilde{\Omega}_i}^e) - \frac{1}{2} I(U; B_{\tilde{\Omega}_i} | S_{\tilde{\Omega}_i}^e)  \Big]\\
&+ \sum_{i:\,\tilde{\Omega}_i \in J_{BC}} \Big[ I(U; C_{\tilde{\Gamma}_i}|T S R_{\tilde{\Gamma}_i}^e) - \frac{1}{2} I(U; C_{\tilde{\Gamma}_i} | R_{\tilde{\Gamma}_i}^e)  \Big]\\
&+ \sum_{i:\,\tilde{\Gamma}_i \in J_{AC}} \Big[ I(U; C_{\tilde{\Gamma}_i}|T S R_{\tilde{\Gamma}_i}^e) - \frac{1}{2} I(U; C_{\tilde{\Gamma}_i} |R_{\tilde{\Gamma}_i}^e)  \Big].
\end{align*}
Next, we rewrite $\chi_O$ by reordering the summands in terms of the indices of boxes and not time-steps:
\begin{align*}
\chi_O = &\sum_{j \in J_{AB}} \Big[ I(U; A_j|T_j^e) - \frac{1}{2} I(U; A_j | T_j^e)  \Big] \\
&+\sum_{j \in J_{AC}} \Big[ I(U; A_j|T_j^e) - \frac{1}{2} I(U; A_j | T_j^e)  \Big] \\
&+ \sum_{j \in J_{AB}} \Big[ I(U; B_j|T S_j^e) - \frac{1}{2} I(U; B_j | S_j^e)  \Big] \\
&+ \sum_{j \in J_{BC}} \Big[ I(U; B_j|T S_j^e) - \frac{1}{2} I(U; B_j | S_j^e)  \Big]\\
&+ \sum_{j \in J_{BC}} \Big[ I(U; C_j|T S R_j^e) - \frac{1}{2} I(U; C_j |R_j^e)  \Big] \\
&+ \sum_{j \in J_{AC}} \Big[ I(U; C_j|T S R_j^e) - \frac{1}{2} I(U; C_j |R_j^e)  \Big].
\end{align*}
Using \eqref{important} and Lemma \ref{B_given_A} (i), for $j \in J_{AB}$ we have
\begin{align}\label{g1}
I(U;& B_j|T S_j^e) - I(U; B_j|T_j^e A_j S_j^e) \nonumber \\
&=I(B_j; T|T_j^e A_j S_j^e U) \geq 0.
\end{align}
Moreover, by \eqref{important} and Lemma~\ref{B_given_A} (ii), for $j \in J_{BC}$ we have
\begin{align}\label{g2}
I(U; B_j|T S_j^e) - I(U; B_j| S_j^e) = I(B_j;T |S_j^e U) \geq 0
\end{align}
We similarly for $j \in J_{BC}$ have
\begin{align}\label{g3}
I(U; C_j|T S R_j^e) - I(U; C_j|S_j^e B_j R_j^e)\nonumber \\
= I(C_j;T S|S_j^e B_j R_j^e U) \geq 0
\end{align}
and for $j \in J_{AC}$ have
\begin{align}\label{g4}
I(U; C_j|T S R_j^e) - I(U; C_j|T_j^e A_j R_j^e)\nonumber \\ = I(C_j;S T|T_j^e A_j R_j^e U) \geq 0
\end{align}
Putting these together we find that
$$\chi_O \geq \chi_{O_1} + \chi_{O_2} + \chi_{O_3}, $$
where,
\begin{align*}
\chi_{O_1}  =&  \sum_{j \in J_{AB}} \Big[ I(U; A_j|T_j^e) + I(U; B_j|T_j^e A_j S_j^e) \\
&- \frac{1}{2} I(U; A_j | T_j^e) -  \frac{1}{2} I(U; B_j | S_j^e)  \Big], \\ \\
\chi_{O_2}  =&  \sum_{j \in J_{AC}} \Big[I(U; A_j|T_j^e) + I(U; C_j|T_j^e A_j R_j^e) \\
&- \frac{1}{2} I(U; A_j | T_j^e) - \frac{1}{2} I(U; C_j |R_j^e) \Big],\\ \\
\chi_{O_3} =&  \sum_{j \in J_{BC}} \Big[I(U; B_j| S_j^e) + I(U; C_j|S_j^e B_j R_j^e) \\
&- \frac{1}{2} I(U; B_j | S_j^e) -\frac{1}{2} I(U; C_j |R_j^e) \Big].\\
\end{align*}
By adding and subtracting $I(U; A_j|T_j^e S_j^e)$ and using~\eqref{important0}  and the chain rule we have 
\begin{align*}
\chi_{O_1}  &=  \sum_{j \in J_{AB}} \Big[ I(U; A_j; S_j^e|T_j^e) + I(U; A_jB_j|T_j^e S_j^e) \\ 
&\qquad - \frac{1}{2} I(U; A_j | T_j^e) -  \frac{1}{2} I(U; B_j | S_j^e)  \Big].
\end{align*}
Next by the data processing inequality we have
\begin{align*}
\chi_{O_1}  &\geq  \sum_{j \in J_{AB}} \Big[ I(U; A_j; S_j^e|T_j^e) \\
&\qquad +\frac12 I(U; A_j|T_j^e S_j^e)+\frac12 I(U; B_j|T_j^e S_j^e) \\ 
&\qquad - \frac{1}{2} I(U; A_j | T_j^e) -  \frac{1}{2} I(U; B_j | S_j^e)  \Big].
\end{align*}
On the other hand, by the no-signaling condition $I(S_j^e; A_j|T_j^e)=0$ we have
\begin{align*}
I(U; A_j|T_j^e S_j^e) &= I(US_j^e; A_j| T_j^e) - I(S_j^e; A_j|T_j^e)\\
&=I(US_j^e; A_j| T_j^e)\\
& = I(U; A_j| T_j^e) + I(S_j^e; A_j| T_j^e U)
\end{align*}
We similarly have
$$I(U; B_j|T_j^e S_j^e) = I(U; B_j| S_j^e) + I(T_j^e; B_j| S_j^e U).$$
Therefore,
\begin{align*}
\chi_{O_1} \geq & \sum_{j \in J_{AB}} \Big[  I(U; A_j; S_j^e |T_j^e )  \\
& \qquad +\frac 12I(S_j^e; A_j| T_j^e U) + \frac12 I(T_j^e; B_j| S_j^e U) \Big]. 
\end{align*}
Next, using~\eqref{important0} and $I(S_j^e; A_j|T_j^e)=0$ we find that
$I(U; A_j; S_j^e |T_j^e ) = -I(A_j; S_j^e| T_j^eU)$.
Putting these together we arrive at
$$2\chi_{O_1} \geq \sum_{j \in J_{AB}} \Big[  I(B_j;T_j^e|S_j^e U)- I(A_j; S_j^e|T_j^e U) \Big] := L_{A\to B}.$$
Following similar computations we also obtain
$$2\chi_{O_2} \geq \sum_{j \in J_{AC}} \Big[ I(C_j;T_j^e|R_j^e U) - I(A_j; R_j^e|T_j^e U)\Big] := L_{A\to C},$$
and
$$2\chi_{O_3} \geq \sum_{j \in J_{BC}} \Big[ I(C_j;S_j^e|R_j^e U)- I(B_j; R_j^e|S_j^e U) \Big] := L_{B\to C}.$$
Hence,
$$2\chi_O \geq L_{A\to B} + L_{A\to C} + L_{B\to C}.$$
If $ L_{A\to B} + L_{A\to C} + L_{B\to C} \geq 0$ the proof is complete. Otherwise, from the beginning we could change the order in which the chain rule in~\eqref{eq:chain-order} is expanded and repeat the same computations. If instead of the order Alice, Bob and Charlie in~\eqref{eq:chain-order} we expand $I(U; TSR)$ in the reverse order Charlie, Bob and Alice we obtain the inequality
$$2\chi_O \geq L_{B\to A} + L_{C\to A} + L_{C\to B}.$$
Now the proof completes once we note that $L_{B\to A} =-L_{A\to B}$ etc.

$\hfill \square$

\section{Tightness}
\label{Appendix_Tightness}

In this appendix we show that the Finner inequalities that we derive in the paper are tight in the following sense.

\begin{theorem}
For any network $\cal{N}$ with parties $\{A_1, \dots, A_n\}$ and sources $\{S_1, \dots, S_m\} $, and arbitrary numbers $0\leq p_j\leq 1$, there exists a fractional independent set $(\eta_1, \dots, \eta_n)$ of $\mathcal N$ and a binary distribution $P_{A_1\dots A_n}\in \NL$ such that 
$P_{A_j}(1) =p_j$ for all $j$ and 
$$P(1, \dots, 1) = \prod_{j=1}^n \left(P_{A_{j}}(1)\right)^{\eta_j}.$$
\end{theorem}

\begin{proof}
Let $(\eta_1^*, \dots, \eta_n^*)$ be a fractional independent set of $\mathcal N$ that optimizes the following linear program:
\begin{align*}
\max & \quad -\sum_{j=1}^n \eta_j\log p_j \\ 
\hbox{s.t.} & \quad \sum_{j: i\to j} \eta_j\le 1 \quad \forall i \qquad \\
& \quad \eta_j\ge 0 \quad\quad  \forall j.
\end{align*}
Consider the dual of this linear program:
\begin{align*}
\min & \quad \sum_{i=1}^m c_i \\ 
\hbox{s.t.} & \quad \sum_{i:i\to j} c_i\ge -\log p_j \quad \forall j\\
& \quad c_i\ge 0 \quad \forall i. 
\end{align*}
Let $(c^*_1, \dots, c^*_m)$ be an optimal solution of this dual program is $c^*$. 
Then by the strong duality of linear programs we have
\begin{align}\label{eq:SDLP}
\sum_i c_i^* = -\sum_j \eta_j^* \log p_j.
\end{align}
Let us split the set of parties in terms of the constraints of the dual linear program:
\begin{align*}
E&= \Big\{j\,\Big|\,  \sum_{i: i\to j} c_i^{*}>-\log p_j\Big\},\\
F&= \Big\{j\,\Big|\,  \sum_{i: i\to j} c_i^{*} = -\log p_j\Big\}.
\end{align*}
For any $j\in E$ and any source $i$ connected to it (i.e., with $i\to j$) pick some $0\leq d^{(j)}_i\leq c_i$ such that 
$$\sum_{i: i\to j} d^{(j)}_i = -\log p_j.$$
Note that since for $j\in E$ we have $ \sum_{i: i\to j} c_i^{*}>-\log p_j$, such $d^{(j)}_i$'s exist. 

Define independent Bernoulli random variables $(X_1, \dots, X_m)$ by 
$$P(X_i=1) = 2^{-c^*_i}.$$
Also for any $j\in E$ and source $i$ connected to it define the binary random variable $Y^{(j)}_{i}$ by
\begin{align*}
P\big(Y_{i}^{(j)}=1\big| X_{i}=1\big)=1,\\
P\big(Y_{i}^{(j)}=1\big)= 2^{-d_i^{(j)}}.
\end{align*}
Observe that such a random variable $Y_{i}^{(j)}$ exists since $d_i^{(j)}\leq c_i$, and that $Y_i^{(j)}$ can be computed given $X_i$ and independent of the rest of random variables. 

Now suppose that the $i$-th source distributes $X_i$. Then party $j\in F$ outputs
$$A_j = \prod_{i: i\to j} X_j,$$
and party $j\in E$ outputs
$$A_j = \prod_{i:i\to j} Y_i^{(j)}.$$
We emphasis once again that $Y_i^{(j)}$ can be computed locally by the $j$-th party having access to $X_i$.  Then by definition the resulting joint distribution $P_{A_1\dots A_n}$ belongs to $\NL$. Also, for every $j\in E$ we have
\begin{align*}
P(A_j=1) = \prod_{i:i\to j} P\big(Y_i^{(j)}=1\big) =\prod_{i:i\to j} 2^{-d_i^{(j)}}
= p_j. 
\end{align*} 
Similarly for every $j\in F$ we have $P(A_j=1)= p_j$. Next we compute $P(1, \dots, 1)$. Note that for every $i$ there exists some $j_i\in F$ with $i\to j$ since otherwise we can decrease $c_i^*$ and improve the objective value of the dual linear program. Then $(A_1, \dots, A_n) = (1, \dots, 1)$, and in particular $A_{j_i}=1$ only if $X_i=1$. On the other hand, by definitions if $X_i=1$ then  $Y_i^{(j)}=1$ for all $j\in E$. We conclude that $(A_1, \dots, A_n) = (1, \dots, 1)$ is equivalent to $X_i=1$ for all $i$. Therefore,
\begin{align*}
P(1, \dots, 1) & = \prod _i P(X_i=1) = \prod_i 2^{-c_i^*} = \prod_j p_j^{\eta_j^*},
\end{align*}    
where we used~\eqref{eq:SDLP}.

\end{proof}

\end{document}